\documentclass{amsart}

\usepackage[american]{babel}

\usepackage[letterpaper,top=2cm,bottom=2cm,left=3cm,right=3cm,marginparwidth=1.75cm]{geometry}

\usepackage{amsmath}
\usepackage{amsthm}
\usepackage{graphicx}
\usepackage[colorlinks=true, allcolors=red]{hyperref}
\usepackage{cite}

\usepackage[american]{babel}	
\usepackage{csquotes}	
\usepackage{graphicx}			
\usepackage{epstopdf}
\usepackage{subfig}				
\usepackage{marvosym}		
\usepackage{amssymb}
\usepackage{mathtools}
\usepackage{color}

\newtheorem{lemma}{Lemma}
\newtheorem{theorem}{Theorem}
\newtheorem{proposition}{Proposition}
\newtheorem{corollary}{Corollary}

\title{Defining binary phylogenetic trees using  parsimony}

\author{Mareike Fischer}

\address{Institute of Mathematics and Computer Science, Greifswald University, Greifswald, Germany} \email{email@mareikefischer.de}

\begin{document}
\maketitle

\begin{abstract} 
Phylogenetic (i.e. leaf-labeled) trees play a fundamental role in evolutionary research. A typical problem is to reconstruct such trees from data like DNA alignments (whose columns are often referred to as characters), and a simple optimization criterion for such reconstructions is maximum parsimony. It is generally assumed that this criterion works well for data in which state changes are rare. In the present manuscript, we prove that each phylogenetic tree $T$ with $n\geq 20 k$ leaves is uniquely defined by the set $A_k(T)$, which consists of  all characters with parsimony score $k$ on $T$. This can be considered as a promising first step towards showing that maximum parsimony as a tree reconstruction criterion is justified when the number of changes in the data is relatively small.
\end{abstract}

\section{Introduction}
Mathematical phylogenetics is concerned with reconstructing the evolutionary (leaf-labeled) trees of species based on data. Typically, the data comes in the form of an alignment (e.g. aligned DNA, RNA or proteins or aligned binary
sequences like absence or presence of certain morphological characteristics), whose columns are also often referred to as characters. While no tree
reconstruction method can guarantee to recover the true tree for all data sets,
it has long been known that in some special cases a tree is uniquely defined by certain alignments. One such example is due to the classic theorem by Buneman \cite{Buneman1971}: If one takes a phylogenetic tree $T$ and regard all its edges as bipartitions of the species set labeling the leaves, one can summarize the resulting binary characters in a set, which we refer to as alignment $A_1(T)$. A re-formulated version of Buneman's theorem states that $A_1(T)$ then uniquely defines $T$ \cite{Fischer2019}. Note that $A_1(T)$ consists of all binary characters which induce precisely one change on $T$ (namely on the edge which they correspond to). 

A natural question arising from this scenario is if $A_k(T)$ defines $T$ for any value $k\geq 1$, i.e. when $k$ changes are allowed rather than just one. It has recently been shown that  $A_2(T)$ indeed defines $T$, but that unfortunately, for all $k\geq 3$, $A_k(T)$ does \emph{not} define $T$ whenever $T$ has $2k$ leaves. In the present note, we will prove the positive result that  $A_k(T)$ always defines $T$ whenever  $n \geq 20k$, where $n$ denotes the number of leaves of $T$.

\section{Preliminaries}
\subsection{Definitions and basic concepts}

We start with some notation. Recall that a \emph{phylogenetic tree} $T=(V,E)$ on a species set $X=\{1, \ldots, n\}$, is a connected acyclic graph with vertex set $V$ and edge set $E$ whose leaves are bijectively labeled by $X$. Such a tree $T$ is also often referred to as phylogenetic $X$-tree. It is called \emph{binary} if all its inner nodes have degree 3. Note that we consider two phylogenetic $X$-trees $T=(V,E)$ and $\widetilde{T}=(\widetilde{V},\widetilde{E})$ to be equal, denoted $T=\widetilde{T}$, if there exists a map $f:V\rightarrow \widetilde{V}$  such that $e=\{u,v\} \in E \Longleftrightarrow \{f(u),f(v)\} \in \widetilde{E}$ and with the additional property that $f(x)=x$ for all $x \in X$. In other words, $f$ is a graph isomorphism which preserves the leaf labelling.

Throughout this manuscript,  when we refer to a tree $T$, we always mean a binary phylogenetic $X$-tree. However, for technical reasons we sometimes also have to consider rooted trees: When an edge is removed from a binary (unrooted) tree, two subtrees remain, both of which have precisely one node of degree 2. This node is considered as the {\em root} of the respective subtree. The two subtrees adjacent to the root of a tree are called \emph{maximum pending subtrees} of that tree, cf. Figure \ref{treedecomp}. 

In the present manuscript, whenever we refer to distances between vertices in a tree $T$, we mean the length of the unique path connecting these vertices in $T$. Moreover, we say that two leaves $v$ and $w$ form a \emph{cherry} $[v,w]$, if $v$ and $w$ are adjacent to the same inner node $u$ of $T$.

Furthermore, recall that a bipartition $\sigma$ of $X$ into two non-empty disjoint subsets $A$ and $B$ is often called an \emph{$X$-split}, and is denoted by $\sigma=A|B$. Recall that there is a  natural relationship between $X$-splits and the edges of a phylogenetic $X$-tree $T$, because the removal of an edge $e$ induces a bipartition $\sigma_e$ of $X$. In the following, the set of all such induced $X$-splits of $T$ will be denoted by $\Sigma(T)$. Recall that for a binary phylogenetic $X$-tree $T$ with $|X|=n$ we have $|\Sigma(T)|=2n-3$ \cite[Prop. 2.1.3]{Semple2003}. Moreover, following \cite{Fischer2021}, the {\em size of an $X$-split $\sigma=A|B$} is defined as $|\sigma|=\min\{|A|,|B|\}$. Given a set of $X$-splits, an element of this set with minimal size is called a {\em minimal split}. 

Moreover, recall that a character $f$ is a function from the taxon set $X$ to a set $\mathcal{C}$ of character states, i.e. $f: X \rightarrow \mathcal{C}$. Note that a finite sequence of characters is also often referred to as \emph{alignment} in biology. While in most biological cases, the order of the characters in an alignment plays an important role, for our purpose it suffices to simply define an alignment as a multiset of characters. In this manuscript, we will only be concerned with \emph{binary characters} (and thus also \emph{binary alignments)}, i.e. without loss of generality $\mathcal{C}= \{a,b\}$. 

There is a close relationship between $X$-splits and binary characters, because every $X$-split can be represented by a binary character by assigning the same state to taxa in the same subset. Throughout this manuscript, we assume for technical reasons and without loss of generality that $f(1)=a$.
If an $X$-split $\sigma_e$ is induced by an edge $e$ of a phylogenetic $X$-tree in the manner explained above, we also say that the corresponding binary character is induced by $e$. 

\par\vspace{0.5cm}
Thus, a binary character $f: X \rightarrow \{a,b\}$ assigns to each leaf of the tree a corresponding state. Now, an \emph{extension} of such a character $f$ on a tree $T$ with vertex set $V$ is a map $g: V \rightarrow \{a,b\}$ such that $g(x)=f(x)$ for all $x \in X$. Moreover, we call $ch(g) = \vert \{ \{u,v\} \in E, \, g(u) \neq g(v)\} \vert$ the \emph{changing number} of $g$ on $T$.

Another concept we need for the present manuscript is the so-called \emph{parsimony score} $l(f,T)$ of a character $f$ on a tree $T$. Here, $l(f,T) = \min\limits_{g} ch(g,T)$, where the minimum runs over all extensions $g$ of $f$ on $T$. The parsimony score of an alignment $A=\{f_1,\ldots, f_m\}$ of characters is then defined as: $l(A,T)=\sum\limits_{i=1}^m l(f_i,T)$. 

Last, for a given tree $T$, we define $A_k(T)$ to be the set consisting of all binary characters $f$ with $l(f,T)=k$. Following \cite{Fischer2019}, we also refer to $A_k(T)$ as the \emph{alignment induced by $T$ and $k$}.

\subsection{Known results}

A basic result that we need throughout this manuscript is the following theorem, which counts the number of characters in $A_k(T)$.

\begin{theorem} \label{thm:lengthAk} \cite{Book_Steel}
Let $T$ be a binary phylogenetic $X$-tree with $|X|=n$. Then, we have:
$$|A_k(T)|=\frac{2n-3k}{k}\binom{n-k-1}{ k-1} \cdot 2^{k-1}.$$
\end{theorem}

Another classic result that will play a fundamental role here is Menger's theorem, which in the context of phylogenetics leads to the following proposition \cite[Lemma 5.1.7 and Corollary 5.1.8]{Semple2003}:

\begin{proposition}\label{menger} \cite[adapted from Corollary 5.1.8]{Semple2003} Let $f$ be a binary character on $X$ employing states from $\mathcal{C}=\{a,b\}$, and let $T$ be a binary phylogenetic $X$-tree. Then $l(f,T)$ is equal to the maximum number of edge-disjoint leaf-to-leaf paths of $T$, where each path connects one leaf in state $a$ with one leaf in state $b$.
\end{proposition}

The next two statements are the classic theorem by Buneman and a direct consequence from it concerning $A_1(k)$.

\begin{theorem}[Buneman theorem \cite{Buneman1971,Semple2003}]\label{buneman} 
 Let $T$ and  $\widetilde{T}$  be two binary phylogenetic $X$-trees. Then,
$T=\widetilde{T}$  if and only if $\Sigma(T)=\Sigma(\widetilde{T})$. 
\end{theorem}

\begin{corollary} \label{bunemanAlignment} Let $T$ and  $\widetilde{T}$  be two binary phylogenetic $X$-trees. Then,
$T=\widetilde{T}$  if and only if $A_1(T)=A_1(\widetilde{T})$.
\end{corollary}

The correctness of Corollary \ref{bunemanAlignment} follows directly from the 1:1 relationship between $\Sigma(T)$ and $A_1(T)$.

Last, we recall the following result from \cite{Fischer2019}, which extends Corollary \ref{bunemanAlignment} to the case $k=2$.

\begin{proposition}[adapted from Proposition 1 in \cite{Fischer2019}]\label{A2good}
Let $T$ and  $\widetilde{T}$  be two binary phylogenetic $X$-trees. Then,
$T=\widetilde{T}$  if and only if $A_2(T)=A_2(\widetilde{T})$.
\end{proposition}

\section{Results}

We are now in a position to state the main result of the present manuscript.

\begin{theorem}\label{characterization}
Let $k \in \mathbb{N}_{\geq 1}$ and let $n \in \mathbb{N}$ such that $n\geq 20k$. Let $T$ and  $\widetilde{T}$  be two binary phylogenetic $X$-trees with $|X|=n$. Then,
$T=\widetilde{T}$  if and only if $A_k(T)=A_k(\widetilde{T})$.
\end{theorem}

Before we can prove Theorem \ref{characterization}, we need to state one more lemma.

\begin{lemma}\label{pathsexistence} Let $T$ be a phylogenetic $X$-tree with $|X|=n$. Then, $T$ has a set of  $\left\lfloor \frac{n}{2}\right\rfloor$ edge-disjoint leaf-to-leaf paths.
\end{lemma}

\begin{proof} Note that if $n<2$, then there is nothing to show. If $n\geq 2$, we apply Theorem \ref{thm:lengthAk} to $k=\left\lfloor \frac{n}{2}\right\rfloor$. For the case where $n$ is even, this leads to: 

$$|A_k(T)|=|A_{\frac{n}{2}}(T)|=\frac{2n-3\frac{n}{2}}{\frac{n}{2}}\binom{n-\frac{n}{2}-1}{ \frac{n}{2}-1} \cdot 2^{\frac{n}{2}-1}=2^{\frac{n}{2}-1}\geq 1.$$

Similarly, if $n$ is odd, we have $n\geq 3$ and $k=\left\lfloor \frac{n}{2}\right\rfloor = \frac{n-1}{2}$ and therefore get 

$$|A_k(T)|=|A_{\frac{n-1}{2}}(T)|=\frac{2n-3\frac{n-1}{2}}{\frac{n-1}{2}} \binom{n-\frac{n-1}{2}-1}{ \frac{n-1}{2}-1} \cdot 2^{\frac{n-1}{2}-1}=\frac{n+3}{n-1}\cdot \frac{n-1}{2} \cdot 2^{\frac{n-1}{2}-1} \geq 1.$$

So in both cases, there is at least one character $f$ on $T$ with parsimony score $k=\left\lfloor \frac{n}{2}\right\rfloor$. However, by Proposition \ref{menger}, this immediately implies that there is also at least one choice of $\left\lfloor \frac{n}{2}\right\rfloor$ edge-disjoint leaf-to-leaf paths in $T$. This completes the proof.
\end{proof}

We are now finally in a position to prove Theorem \ref{characterization}, which is the main result of the present note.

\begin{proof}[Proof of Theorem \ref{characterization}] Note that the cases $k=1$ and $k=2$ are already proved by Corollary \ref{bunemanAlignment} and Proposition \ref{A2good}. Therefore, we may assume in the following that $k\geq 3$.

Now, let $T\neq \widetilde{T}$ be two phylogenetic $X$-trees with $n$ leaves such that $n \geq 20k$. Then $\Sigma(T)\neq \Sigma(\widetilde{T})$ by Theorem \ref{buneman}, and as explained above we have $|\Sigma(T)|=|\Sigma(\widetilde{T})|=2n-3$. Together, this implies that $\Sigma(T)\setminus \Sigma(\widetilde{T}) \neq \emptyset$. Let $\sigma=A|B \in \Sigma(T)\setminus \Sigma(\widetilde{T})$ be minimal, i.e. of minimal size. Without loss of generality, we assume $|A|=|\sigma|$, i.e. $|A|\leq |B|$ and thus $|B|\geq \frac{n}{2}$. Note that $|A|\geq 2$ as $\sigma \in \Sigma(T)$ but $\sigma\not\in \Sigma(\widetilde{T})$ (otherwise, $\sigma$ would be contained in both split sets as all $X$-trees contain edges leading to each of the leaves in $X$). Moreover, $\sigma$ divides $T$ into two subtrees $T_A$ with leaf set $A$ and $T_B$ with leaf set $B$. In the following, we denote by $T_{A_1}$ and $T_{A_2}$ the two maximal pending subtrees of $T_A$, which must exist as $|A|\geq 2$, cf. Figure \ref{treedecomp}. The taxon sets of $T_{A_1}$ and $T_{A_2}$ are denoted by $A_1$ and $A_2$, respectively. Note that $|A_1|<|A|$ and $|A_2|<|A|$, and also note that $\Sigma(T)$ must contain the two $X$-splits $\sigma_1=A_1|X\setminus A_1$ and $\sigma_2=A_2|X\setminus A_2$ as $T_{A_1}$ and $T_{A_2}$ are subtrees of $T$.

\begin{figure} 
\center
\scalebox{.17}{\includegraphics{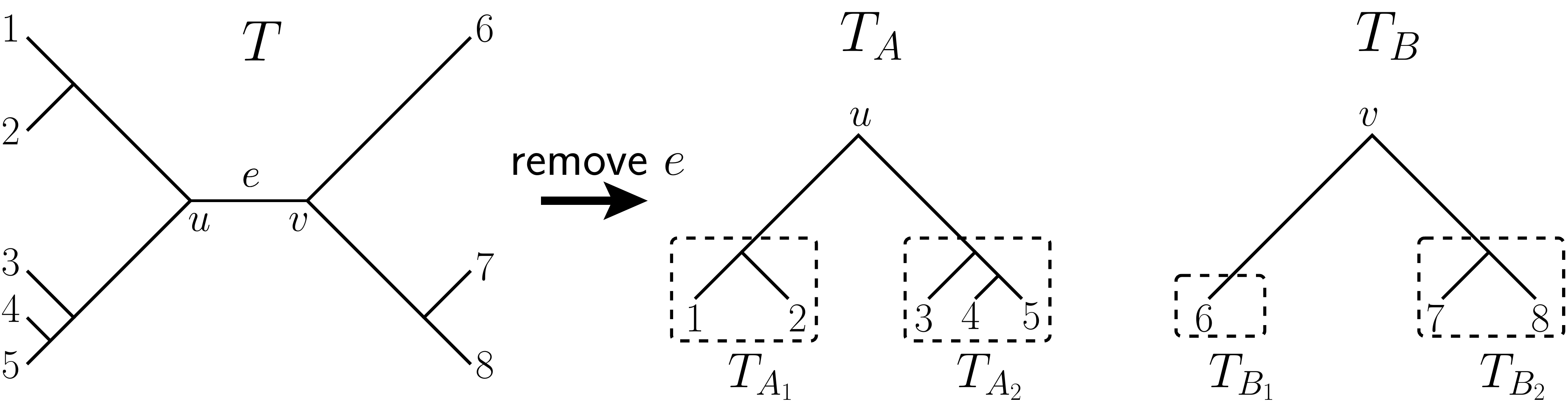}}

\caption{  \scriptsize (taken from \cite{Fischer2019}) By removing an edge $e$ from an unrooted phylogenetic tree $T$, it is decomposed into two rooted subtrees, $T_A$ and $T_B$. If, as in this figure, both of them consist of more than one node, then we can further decompose them into their two maximal pending subtrees, $T_{A_1}$ and $T_{A_2}$ or $T_{B_1}$ and $T_{B_2}$, respectively. }
\label{treedecomp}
\end{figure}

Note that by the minimality of  $\sigma$, the $X$-splits $\sigma_1$ and $\sigma_2$ must be contained in $\Sigma(T)\cap\Sigma(\widetilde{T})$. 
So we have $\sigma_1$, $\sigma_2$ $\in \Sigma(\widetilde{T})$. By an analogous argument, $\widetilde{T}$ must also contain all splits induced by edges of $T_{A_1}$ and $T_{A_2}$, respectively. So in fact, $\widetilde{T}$ has $T_{A_1}$ and $T_{A_2}$ as subtrees, but as $\sigma \not\in \Sigma(T)\cap\Sigma(\widetilde{T})$, they are not pending on the same inner node of  $\widetilde{T}$. In particular, this implies that $\widetilde{T}$ can be represented as depicted in Figure \ref{tildeTfig}. Moreover, as $\widetilde{T}$ does not contain split $\sigma$, it must contain subtrees $\widetilde{T}_{B_1},\ldots , \widetilde{T}_{B_m}$, where $m\in \mathbb{N}_{\geq 2}$, along the path from the root $\rho_1$ of $T_{A_1}$ to the root $\rho_2$ of $T_{A_2}$ (note that all these subtrees contain at least one leaf). Without loss of generality, we assume that $\widetilde{T}_{B_1}$ is pending on the same node as $T_{A_1}$ and $\widetilde{T}_{B_m}$ is pending on the same node as $T_{A_2}$, cf. Figure \ref{tildeTfig}. In the following, we denote by $B_i$ the set of leaves of subtree $\widetilde{T}_{B_i}$ for all $i=1,\ldots,m$.

\begin{figure} 
\center

\scalebox{.5}{\includegraphics{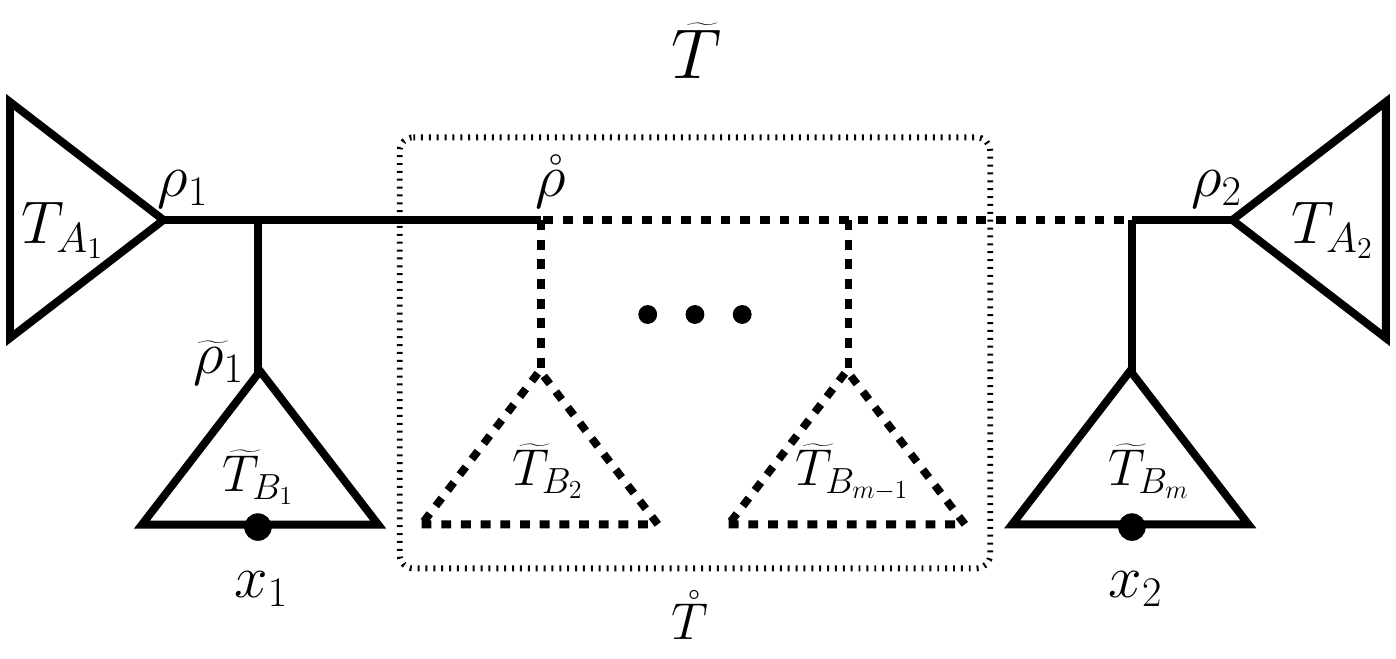}}
\caption{  \scriptsize  Tree $\widetilde{T}$ as described in the proof of Theorem \ref{characterization}. Note that the subtree $\mathring{T}$ might be empty, namely if $m=2$. If $m=3$, the proof considers $\widetilde{T}_{B_2}$, but if $m>3$, we can think of $\mathring{T}$ as a rooted binary tree with root $\mathring{\rho}$ and with $\widetilde{b}$ leaves as described in the proof. }
\label{tildeTfig}
\end{figure}

Now we construct a binary character $f$ such $l(f,\widetilde{T})=k>l(f,T)$ as follows: \begin{itemize}
\item We assign state $a$ to all taxa in $A$. 
\item Then we choose two paths: One connecting the root $\rho_1$ of $T_{A_1}$ with the nearest leaf $x_1\in \widetilde{T}_{B_1}$, and the other one connecting the root $\rho_2$ of $T_{A_2}$ with the nearest leaf $x_2 \in \widetilde{T}_{B_m}$.
\item Then we choose $k-2$ additional edge-disjoint leaf-to-leaf paths in $\widetilde{T}$, where all leaves considered are contained in $B\setminus\{x_1,x_2\}$.
\item We then assign $x_1$ and $x_2$ state $b$. For all other $k-2$ paths, we label one of its endpoints $a$ and the other one $b$. All unlabeled leaves in $B$ that are not contained in a path (if such leaves exist) also are assigned state $b$.
\end{itemize}

We subsequently prove  $l(f,\widetilde{T})=k>l(f,T)$, but  before we do this, we first show that our construction is valid, i.e. that indeed $k-2$ paths can be chosen as described above in the crucial third step.

Let $p_1$ denote the path length  of the path from $\rho_1$ to $x_1$ in $\widetilde{T}$ and similarly, let $p_2$ denote the path length of the path from $\rho_2$ to $x_2$ in $\widetilde{T}$. Note that these two paths are edge-disjoint. 

Now consider $\widetilde{T}_{B_1}$ with its $|B_1|$ leaves and root $\widetilde{\rho}_1$: Before taking the path from $\widetilde{\rho}_1$ to $x_1$ (as a subpath of the path from $\rho_1$ to $x_1$), by Lemma \ref{pathsexistence}, there were $\left\lfloor \frac{|B_1|}{2}\right\rfloor$ edge-disjoint leaf-to-leaf paths in $\widetilde{T}_{B_1}$. Even without counting the number of such paths requiring edges of the already taken path, one can easily see that this path can at most reduce the number of edge-disjoint leaf-to-leaf paths in $\widetilde{T}_{B_1}$ by $p_1-2$ (as it contains $p_1-2$ edges in $\widetilde{T}_{B_1}$). 
 So in total, subtree $\widetilde{T}_{B_1}$ allows for at least $\left\lfloor \frac{|B_1|}{2}\right\rfloor - p_1+2$ edge-disjoint leaf-to-leaf paths. By the same argument, $\widetilde{T}_{B_m}$ allows for at least $\left\lfloor \frac{|B_m|}{2}\right\rfloor - p_2+2$ edge-disjoint leaf-to-leaf paths.

We now consider the leaves in $B\setminus (B_1\cup B_m)$. Let $\widetilde{b}$ denote the cardinality of this set. Note that $\widetilde{b}$ can equal 0 (namely if $m=2$). On the other hand, if $\widetilde{b}=1$ (i.e. if $m=3$), we can consider $\widetilde{T}_{B_2}$, and otherwise, if $\widetilde{b}\geq 2$ (i.e. $m\geq 4$) we can consider the rooted tree $\mathring{T}$ containing subtrees $\widetilde{T}_{B_2},\ldots, \widetilde{T}_{B_{m-1}}$ as depicted in the dashed box of Figure \ref{tildeTfig}. Again by Lemma \ref{pathsexistence}, there are in all cases at least $\left\lfloor \frac{\widetilde{b}}{2}\right\rfloor$ edge-disjoint leaf-to-leaf paths induced by $B\setminus (B_1\cup B_m)$ in the tree under consideration (i.e. the empty tree, $\widetilde{T}_{B_2}$ or $\mathring{T}$, respectively). Note that each such path naturally corresponds to a path in $\widetilde{T}$ which is edge-disjoint with all other paths chosen before.

So the total number $\mathcal{P}$ of edge-disjoint leaf-to-leaf paths that are present in in $B$ which also do not intersect with the paths from $\rho_1$ to $x_1$ and $\rho_2$ to $x_2$, respectively, is bounded as follows: 

\begin{equation*} \mathcal{P} \geq \left\lfloor \frac{|B_1|}{2}\right\rfloor - (p_1-2) + 
\left\lfloor \frac{|B_m|}{2}\right\rfloor - (p_2-2)+ \left\lfloor \frac{\widetilde{b}}{2}\right\rfloor \geq \left( \frac{|B_1|-1}{2}\right) + \left( \frac{|B_m|-1}{2}\right) + \left(\frac{\widetilde{b}-1}{2}\right) -(p_1-2)-(p_2-2).
\end{equation*}

Note that by the choice of $x_1$ the height of $\widetilde{T}_{B_1}$ is at least $p_1-2$ and that in fact \emph{all} leaves of $\widetilde{T}_{B_1}$ have at least this distance to the root  $\widetilde{\rho}_1$ of $\widetilde{T}_{B_1}$ (as the distance from $x_1$ to $\widetilde{\rho}_1$ and to $\rho_1$ is minimal). So this leads to $|B_1| \geq 2^{p_1-2}$ and thus $p_1-2 \leq \log_2|B_1|$. Analogously, we derive $p_2-2 \leq \log_2|B_m|$. Using this in the above inequality leads to: 

\begin{equation*}\label{Pbound2} \mathcal{P} \geq \left( \frac{|B_1|-1}{2}\right) + \left( \frac{|B_m|-1}{2}\right) + \left(\frac{\widetilde{b}-1}{2}\right) -\log_2|B_1|- \log_2|B_m|= \frac{|B|}{2} - \log_2|B_1|- \log_2|B_m| -1.5,
\end{equation*}

where the last equation is due to $|B_1|+|B_m|+\widetilde{b}=|B|$. Next we use the fact that $B_1$ and $B_m$ are both proper subsets of $B$ and thus we have $\log_2|B_1| \leq \log_2|B|$ as well as $\log_2|B_m| \leq \log_2|B|$. This leads to: 

\begin{equation*}\label{Pbound3} \mathcal{P} \geq  \frac{|B|}{2} -2 \log_2|B|-1.5.
\end{equation*}

As stated above, we wish to select $k-2$ edge-disjoint leaf-to-leaf paths in $B\setminus\{x_1,x_2\}$, and we have shown that at least $\frac{|B|}{2} -2 \log_2|B|-1.5$ such paths. So it remains to show that this number is at least $k-2$ for $k\geq 3$ and  $n\geq 20k$.

In this regard, we now set $g(y):=\frac{1}{2}y-2 \log_2(y) -1.5$ and analyze this function. Using the first derivative $g'(y)=\frac{1}{2}-\frac{2}{y \ln (2)}$, which equals 0 if and only if $y= \frac{4}{\ln (2)} \approx 5.77$, as well as the second derivate $g''(y)=\frac{2}{y^2 \ln (2)}>0 \ \ \forall y \neq 0$, it can be easily seen that $g$ has a local minimum at $y \approx 5.77$ and that for all values of $y$ larger than this minimum, $g$ is strictly monotonically increasing. Additionally, for $y=30$ we have $g(y)\approx 3.686>3= \frac{y}{10}$, so by the monotonicity of $g$ and as $g'(y)>\frac{1}{10}=\left(\frac{y}{10}\right)'$ for all $y\geq 8$, we can conclude for all $|B|\geq 30$:

\begin{equation*}\label{Pbound4} \mathcal{P} \geq  \frac{|B|}{2} -2 \log_2|B|-1.5 > \frac{|B|}{10} \geq \frac{n}{20}\geq  \frac{20k}{20}=k>k-2 \ \ \forall k\geq 3, \ n\geq 20k. 
\end{equation*}

Note that the third inequality is due to $|B|\geq \frac{n}{2}$.  Moreover, note that as $k \geq 3$, we have $n \geq 20k \geq  60$, which guarantees that $|B|\geq 30$, so that the latter requirement is no restriction. Thus, it is indeed possible to choose $k-2$ edge-disjoint leaf-to-leaf paths in $B$, additional to the two paths from $\rho_1$ to $x_1$ and $\rho_2$ to $x_2$, respectively. 

Recall that we labeled all taxa in $A$ with $a$, and $x_1$ and $x_2$ were labeled $b$. Moreover, we now have chosen $k-2$ edge-disjoint leaf-to-leaf paths, and for each such path, we label one of its leaves $a$ and the other one $b$. All other taxa not covered by such paths (if they exist) are labeled with $b$. We call the resulting character $f$. 

Next, we show that $l(f,\widetilde{T})=k$. Let us pick one leaf $a_1$ from $T_{A_1}$ and one leaf $a_2$ from $T_{A_2}$ and consider the paths from $a_1$ to $x_1$ and $a_2$ to $x_2$, respectively. Together with the $k-2$ paths chosen from $B$, this leads to $k$ edge-disjoint paths connecting leaves in state $a$ with leaves in state $b$. So, by Proposition \ref{menger}, we have $l(f,\widetilde{T})\geq k$.

On the other hand, as all leaves in the subtrees $T_{A_1}$ and $T_{A_2}$ are in state $a$, $f$ requires no substitutions in these subtrees. In fact, if we replace $T_{A_1}$ and $T_{A_2}$ by leaves $a_1$ and $a_2$ and call the resulting tree $\widehat{T}$ and the resulting restriction of $f$ on the remaining leaves $\widehat{f}$, we have $l(f,\widetilde{T})=l(\widehat{f},\widehat{T})$. However, as $\widehat{f}$ has precisely $k$ leaves in state $a$, the parsimony score of $\widehat{f}$ on any tree can be at most $k$ (a change might be needed on all pending edges leading to the $a$-leaves, but more changes are definitely not required). Therefore, we obtain 
$l(\widehat{f},\widehat{T})\leq k$ and thus also $l(f,\widetilde{T})\leq k$.

Altogether we conclude that $l(f,\widetilde{T})= k$ and therefore $f \in A_k(\widetilde{T})$.

Moreover, we now argue that $l(f,T)<k$ and that therefore $f \not\in A_k(T)$. In order to see this, we again replace the subtrees $T_{A_1}$ and $T_{A_2}$ by leaves $a_1$ and $a_2$, respectively, and we again obtain character $\widehat{f}$, this time in combination with the restriction of $T$ on the smaller leaf set, which we will call $T'$. Now as above, we have $l(f,T)=l(\widehat{f},T')$, and as there are only $k$ leaves labeled $a$, we conclude as above that $l(f,T)\leq k$. However, in $T'$, the leaves $a_1$ and $a_2$ form a cherry $[a_1,a_2]$, so no extension minimizing the changing number to give $l(\widehat{f},T')$ can require a change for both of these leaves. (If node $u$ incident to both $a_1$ and $a_2$ was in state $b$, there would be two changes on the cherry, but then it would be advantageous to change $u$ to $a$, as in this case, while there might be an additional change on the edge corresponding to $\sigma$, two changes, namely on the edges $\{u,a_1\}$ and $\{u,a_2\}$, could be saved.) Therefore, in total at most $k-1$ changes are needed, so in fact, we have $l(\widehat{f},T')<k$ and thus also $l(f,T)<k$, which shows that $f \not\in A_k(T)$.

In summary, we have found a character $f \in A_k(\widetilde{T})$, for which we know that $f \not\in A_k(T)$, which shows that $A_k(\widetilde{T})\neq A_k(T)$. This completes the proof.
\end{proof}

\section{Discussion and outlook}
The main result of this manuscript, namely that all binary phylogenetic trees with $n\geq 20k$ leaves are uniquely defined by their induced profiles $A_k(T)$, can be seen as an extension of Corollary \ref{bunemanAlignment}, which is a consequence of the classic Buneman theorem \ref{buneman} \cite{Buneman1971}, as well as Proposition \ref{A2good}. Based on these results as well as based on the computational considerations from \cite{pablo}, we conjecture that the factor of 20 in Theorem \ref{characterization} can be reduced. However, note that it was already shown in \cite{Fischer2019} that a factor of 2 is not sufficient whenever $k\geq 3$ and $n=2k$. In any case, studying the gap between 2 and 20 is an interesting area for further investigations.

Moreover, note that in \cite{Fischer2019}, the fact that $A_2(T)$ defines $T$ was only used as a first step (namely, a necessary pre-requisite) to show that $T$ can also be recovered when using maximum parsimony as a criterion for tree reconstruction. In particular, it was shown there that for all $n \geq 9$, $T$ is the so-called unique maximum parsimony tree for $A_2(T)$, i.e. $T=T'$, where $T'$ is the tree minimizing $l(A_2(T),T')$. It this regard, we consider the present manuscript as an important step to answer the same question for $A_k(T)$ whenever $n \geq 20k$: Can $T$ be uniquely recovered from $A_k(T)$ when maximum parsimony is used for tree reconstruction? This question is mathematically intriguing, but also relevant for biologists, as maximum parsimony as a simple tree reconstruction criterion is often considered valid whenever the number of changes is relatively small (cf. e.g. \cite{Lin1991,Sourdis}). But of course whenever $A_k(T)$ does not uniquely define $T$, no tree reconstruction method will be able to uniquely recover $T$ from $A_k(T)$, which highlights the importance of characterizing cases when $A_k(T)$ indeed does characterize $T$. Theorem \ref{characterization} of the present manuscript represents an essential step in this regard.

\section*{Acknowledgements} I want to thank Mike Steel for helpful discussions on a previous version of this manuscript. Moreover, want to thank the German Academic Exchange Service DAAD for funding a conference trip to New Zealand in 2019, where I was first inspired to start working on this project. The rest of this project was completed as part of the joint research project \textit{\textbf{DIG-IT!}}, which is kindly supported by the European Social Fund (ESF), reference: ESF/14-BM-A55-0017/19, and the Ministry of Education, Science and Culture of Mecklenburg-Vorpommern, Germany. Last but not least, I want to thank two anonymous reviewers, whose valuable suggestions helped to improve this manuscript.

\section*{Conflict of interest} The author herewith certifies that she has no affiliations with or involvement in any
organization or entity with any financial (such as honoraria; educational grants; participation in speakers’ bureaus;
membership, employment, consultancies, stock ownership, or other equity interest; and expert testimony or patent-licensing
arrangements) or non-financial (such as personal or professional relationships, affiliations, knowledge or beliefs) interest in the subject matter discussed in this manuscript.

\section*{Data availability statement} 
Data sharing is not applicable to this article as no new data were created or analyzed in this study.

\bibliographystyle{plain}
\bibliography{References-NEW}   

\end{document}